\documentclass{article}

\def\noheaderplainsetup{

\topmargin=0pt \headheight=0pt \headsep=0pt  \oddsidemargin=0pt \evensidemargin=0pt  \textheight=9.1truein \textwidth=6.5truein}   

\noheaderplainsetup

\usepackage{amsfonts}    

\begin{document}

\newcommand{\Rank}{\mbox{Rank}}
\newcommand{\rank}{\mbox{\scriptsize Rank}}
\newcommand{\lll}{\mbox{{\bf CL16}}}
\newcommand{\mmm}{\mbox{{\bf CL17}}}
\newcommand{\Bigmlc}{\mbox{{\Large $\wedge$}}}
\newcommand{\Bigmld}{\mbox{{\Large $\vee$}}}
\newcommand{\bigmlc}{\mbox{{\large $\wedge$}}}
\newcommand{\bigmld}{\mbox{{\large $\vee$}}}
\newcommand{\bigleftbrace}{\mbox{{\large $\{$}}}
\newcommand{\bigrightbrace}{\mbox{{\large $\}$}}}
\newcommand{\Bigleftbrace}{\mbox{{\Large $\{$}}}
\newcommand{\Bigrightbrace}{\mbox{{\Large $\}$}}}
\newcommand{\emptyrun}{\langle\rangle}
\newcommand{\legal}[2]{\mbox{\bf Lr}^{#1}_{#2}} 
\newcommand{\win}[2]{\mbox{\bf Wn}^{#1}_{#2}} 
\newcommand{\seq}[1]{\langle #1 \rangle}           


\newcommand{\ade}{\mbox{\large $\sqcup$}}      
\newcommand{\ada}{\mbox{\large $\sqcap$}}      
\newcommand{\gneg}{\neg}                  
\newcommand{\mli}{\rightarrow}                     
\newcommand{\mld}{\vee}    
\newcommand{\mlc}{\wedge}  
\newcommand{\add}{\hspace{0pt}\sqcup}           
\newcommand{\adc}{\hspace{0pt}\sqcap}
\newcommand{\st}{\mbox{\raisebox{-0.05cm}{$\circ$}\hspace{-0.13cm}\raisebox{0.16cm}{\tiny $\mid$}\hspace{2pt}}}
\newcommand{\cost}{\mbox{\raisebox{0.12cm}{$\circ$}\hspace{-0.13cm}\raisebox{0.02cm}{\tiny $\mid$}\hspace{2pt}}}


\newtheorem{theoremm}{Theorem}[section]
\newtheorem{factt}[theoremm]{Fact}
\newtheorem{definitionn}[theoremm]{Definition}
\newtheorem{lemmaa}[theoremm]{Lemma}
\newtheorem{conventionn}[theoremm]{Convention}
\newtheorem{claimm}[theoremm]{Claim}
\newtheorem{corollaryy}[theoremm]{Corollary}
\newtheorem{examplee}[theoremm]{Example}
\newtheorem{remarkk}[theoremm]{Remark}

\newenvironment{definition}{\begin{definitionn} \em}{ \end{definitionn}}
\newenvironment{theorem}{\begin{theoremm}}{\end{theoremm}}
\newenvironment{lemma}{\begin{lemmaa}}{\end{lemmaa}}
\newenvironment{fact}{\begin{factt}}{\end{factt}}
\newenvironment{corollary}{\begin{corollaryy}}{\end{corollaryy}}
\newenvironment{claim}{\begin{claimm}}{\end{claimm}}
\newenvironment{convention}{\begin{conventionn} \em}{\end{conventionn}}
\newenvironment{proof}{ {\bf Proof.} }{\  $\Box$ \vspace{.1in} }
\newenvironment{example}{\begin{examplee} \em}{\end{examplee}}
\newenvironment{remark}{\begin{remarkk} \em}{\end{remarkk}}

\title{Elementary-base cirquent calculus II: Choice quantifiers}
\author{Giorgi Japaridze
  \\  
 \\ Villanova University and\\ Institute of Philosophy, Russian Academy of Sciences\\
 Email: giorgi.japaridze@villanova.edu\\
 URL: http://www.csc.villanova.edu/$^\sim$japaridz/
}
\date{}
\maketitle

\begin{abstract} Cirquent calculus is a novel proof theory  permitting component-sharing between  logical expressions.   
Using it, the predecessor article ``Elementary-base cirquent calculus I: Parallel and choice connectives'' built the sound and complete axiomatization $\lll$  of a propositional fragment of computability logic. The atoms of the language of $\lll$ represent elementary, i.e., moveless, games, and the logical vocabulary consists of negation, parallel connectives and choice connectives. The present paper constructs the first-order version $\mmm$ of $\lll$, also enjoying soundness and completeness. The language of $\mmm$ augments that of $\lll$ by including choice quantifiers. Unlike classical predicate calculus, $\mmm$ turns out to be decidable.  
\end{abstract}

\noindent {\em MSC}: primary: 03B47; secondary: 03B70; 03F03; 03F20; 68T15. 

\  

\noindent {\em Keywords}: Proof theory; Cirquent calculus; Resource semantics; Deep inference; Computability logic

\section{Preface}\label{intr}

 Cirquent calculus  is a family of deep inference (cf. \cite{gug}) proof systems permitting various sorts of component-sharing between different parts of logical expressions.  
The earlier article \cite{cl16} constructed a cirquent calculus   system $\lll$ for the elementary-base, recurrence-free propositional fragment of computability logic (the game-semantically conceived logic of computational problems introduced in \cite{Jap03})  and proved its   soundness and completeness. The present article takes that result to the first-order level, with the so called choice quantifiers. The resulting system  $\mmm$, in contrast to its classical counterpart, is a decidable predicate logic. While a variety of propositional systems have been built by now in cirquent calculus \cite{Cirq,Japdeep,taming1,taming2,cl16,XuIGPL,XuIf,XuLast}, $\mmm$ is the first cirquent calculus  system  with quantifiers. The formal semantics for its language (and beyond) was set up in \cite{lmcs}, but no axiomatizations had been attempted so far. 

Being a continuation  of \cite{cl16} in the proper sense, the present article should  only be read in combination with its predecessor, as it relies on but does not reintroduce the main concepts from \cite{cl16}. Nor does it discuss related literature or the relevant motivations and philosophy underlying  cirquent calculus and computability logic, as this, again, is done in \cite{cl16}. 

The language of $\mmm$ extends that of $\lll$ by augmenting its logical vocabulary with the {\em choice universal quantifier} $\ada$ (``{\em chall}'') and {\em choice existential quantifier} $\ade$ (``{\em chexists}''), and allowing atoms of any arities. Throughout this article, for simplicity, we assume that the universe of discourse is always the set $\mathbb{N}$ of natural numbers, by innocent abuse of concepts identified with the corresponding decimal numerals. If so, $\ada xG(x)$ is the game where, at the beginning of a play, the environment chooses one of $n\in\mathbb{N}$, after which the game continues as $G(n)$; if such a choice is never made, then the environment loses.  $\ade xG(x)$ is similar, only here it is the machine who can and must make an initial  choice. Thus, $\ada xG(x)$ is in fact nothing but the infinite choice conjunction $G(0)\adc G(1)\adc G(2)\adc\cdots$, and $\ade xG(x)$ 
is the infinite choice disjunction $G(0)\add G(1)\add G(2)\add\cdots$. 

\section{Syntax}\label{Syntax}

The  language of $\mmm$ is the same as that of first-order classical logic, only with the quantifiers $\forall,\exists$ replaced by their choice (``constructive'') counterparts $\ada,\ade$, and additionally including the choice connectives $\adc,\add$ as well as the decimal numerals $0,1,2,\cdots$ referred to as {\bf constants}.   
We will be typically using $x,y,z,\cdots$ as metavariables for the variables of the language, $p,q,r,\cdots$ for its predicate letters,  and $a,b,c\cdots$ for its constants.

We fix four pairwise disjoint infinite sets $\mathbb{C}(\add)$, $\mathbb{C}(\adc)$,  $\mathbb{C}(\ade)$ and $\mathbb{C}(\ada)$, whose elements will be referred to as   {\bf $\add$-clusters},  {\bf $\adc$-clusters}, {\bf $\ade$-clusters} and  {\bf $\ada$-clusters}, respectively. The Gothic letters $\mathfrak{a}, \mathfrak{b}, \mathfrak{c},\cdots$ will be used as metavariables for clusters.  

A {\bf term} is either a variable  or a constant.  A {\bf nonlogical  atom} is $p(t_1,\cdots,t_n)$, where $p$ is an $n$-ary predicate letter and $t_1,\cdots,t_n$ are terms. A {\bf nonlogical literal} is either $A$  or $\neg A$, where $A$ is a nonlogical atom.  The expressions $\top$ and $\bot$ are said to be {\bf logical literals}. 
 A (choice) {\bf universal quantor} (resp. {\bf existential quantor}) is the expression $\ada^\mathfrak{c}$ (resp. $\ade^\mathfrak{c}$), where $\mathfrak{c}$ is a $\ada$- (resp. $\ade$-) cluster.

\begin{definition}\label{basecir}
A {\bf cirquent} is defined inductively as follows:\vspace{-5pt}
\begin{itemize}
\item Each  (logical or nonlogical)  literal is a cirquent.\vspace{-7pt}
\item If $A$ and $B$ are cirquents, then $(A)\mlc(B)$ is a cirquent.\vspace{-7pt}   
\item If $A$ and $B$ are cirquents, then $(A)\mld(B)$ is a cirquent.\vspace{-7pt}
\item If $A$ and $B$ are cirquents and $\mathfrak{c}$ is a $\adc$-cluster,  then $(A)\adc^\mathfrak{c}(B)$ is a cirquent.\vspace{-7pt}
\item If $A$ and $B$ are cirquents and $\mathfrak{c}$ is a $\add$-cluster,  then $(A)\add^\mathfrak{c}(B)$ is a cirquent.\vspace{-7pt}
\item If $A$ is a cirquent, $x$ is a variable and $\mathfrak{c}$ is a $\ada$-cluster,  then $\ada^\mathfrak{c}x(A)$ is a cirquent.\vspace{-7pt}
\item If $A$ is a cirquent, $x$ is a variable and $\mathfrak{c}$ is a $\ade$-cluster,  then $\ade^\mathfrak{c}x(A)$ is a cirquent.
\end{itemize}   
\end{definition}

A cirquent of the form $(A)\mlc(B)$ is said to be {\bf $\mlc$-rooted}, a cirquent of the form $A\adc^\mathfrak{c}B$ is said to be $\adc^\mathfrak{c}$-rooted, and similarly for $\mld,\add^\mathfrak{c},\ada^\mathfrak{c},\ade^\mathfrak{c}$. If we simply say ``$\adc$-rooted'', it should be understood as ``$ \adc^\mathfrak{c}$-rooted for whatever $ \mathfrak{c}$''. Similarly for $\add,\ada,\ade$.

As in \cite{cl16}, negation is only allowed to be applied to atoms. $\neg \ada^\mathfrak{c} x F$ should be understood as an abbreviation of $\ade^\mathfrak{c} x\neg F$, and $\neg \ade^\mathfrak{c} xF$ as an abbreviation of $\ada^\mathfrak{c} x\neg F$. All other conventions of \cite{cl16} regarding the usage of $\neg$ and $\mli$ remain in force.  

When omitting parentheses in cirquents,  our convention is that $\neg$, $\ada^\mathfrak{c}x$ and $\ade^\mathfrak{c}x$ have the highest precedence, then comes $\mli$, then $\adc^\mathfrak{c}$ and $\add^\mathfrak{c}$, and then $\mlc$ and $\mld$.   

All (other) standard terminological and notational conventions of traditional logic also remain in force. This includes the concepts of free and bound occurrences of variables, or the practice of representing a cirquent as $A(x)$ when first mentioning it, and then writing $A(t)$ to mean the result of replacing in $A(x)$ all free occurrences of the variable $x$ by the term $t$. A cirquent is said to be {\bf closed} iff it has no free occurrences of variables. 

\section{Semantics} 

We (re)define {\em LegalRuns} as the set of all runs $\Gamma$ satisfying 
the following conditions:\vspace{-5pt}
\begin{enumerate}
\item Every move of $\Gamma$ is the string $\mathfrak{c}.a$, where either (1) $\mathfrak{c}$ is a $\adc$- or $\add$-cluster and $a\in\{0,1\}$, or (2) $\mathfrak{c}$ is a $\ada$- or $\ade$-cluster and $a\in\mathbb{N}$.\vspace{-7pt}

\item Whenever $\Gamma$ contains a move $\mathfrak{c}.a$ where $\mathfrak{c}$ is a $\adc$- or $\ada$-cluster,  the move is $\bot$-labeled.\vspace{-7pt}   

\item  Whenever $\Gamma$ contains a move $\mathfrak{c}.a$ where $\mathfrak{c}$ is a $\add$- or $\ade$-cluster, the move is $\top$-labeled.\vspace{-7pt} 

\item  For any cluster $\mathfrak{c}$,  $\Gamma$ contains at most one move of the form $\mathfrak{c}.a$.\vspace{5pt}
 \end{enumerate}

The intuitive meaning of condition 1 is that every move signifies either a choice between ``left'' ($0$) and ``right'' ($1$) in a $\adc$- or $\add$-cluster, or a choice among the constants $ 0,1,2,\cdots $ in some $\ada$- or $\ade$-cluster. Conditions 2 and 3 say that the environment moves (chooses) only in $\adc$- or $\ada$-clusters, and the machine only in $\add$- or $\ade$-clusters. Finally,  condition 4 says that, in any given cluster, a choice can be made only once. 

As in \cite{cl16}, given a run $\Gamma\in LegalRuns$, we  say that a cirquent of the form $A\add^\mathfrak{c} B$ or $A\adc^\mathfrak{c} B$ is {\bf $\Gamma$-resolved} iff 
$\Gamma$ contains one of the moves $\mathfrak{c}.0$ or $\mathfrak{c}.1$; then by the {\bf $\Gamma$-resolvent} of the cirquent we mean $A$ if such a move is  
$\mathfrak{c}.0$, and $B$ if it is $\mathfrak{c}.1$. Extending this terminology to quantifiers, we say that a cirquent of the form $\ade^\mathfrak{c}x A(x)$ or $\ada^\mathfrak{c} xA(x)$ is {\bf $\Gamma$-resolved} iff 
$\Gamma$ contains the move $\mathfrak{c}.a$ for some constant $a$; then by the {\bf $\Gamma$-resolvent} of such a cirquent we mean $A(a)$. Sometimes, instead of saying that the cirquent $A\add^\mathfrak{c} B$ (or $A\adc^\mathfrak{c} B$, or $ \ada^\mathfrak{c}x A(x)$, or $\ade^\mathfrak{c}x A(x)$) is resolved, we may simply say that the cluster $\mathfrak{c}$ is resolved. In all cases, as expected, ``{\bf  unresolved}'' means ``not resolved''. 

An {\bf interpretation}  is a function $^*$ that assigns   
to every $n$-ary predicate letter $p$ a relation $p^*\subseteq\mathbb{N}^n$. When $(c_1,\cdots,c_n)\in p^*$, we say that $^*$ makes the atom $p(c_1,\cdots,c_n)$ {\bf true}, or simply that   $p^*(c_1,\cdots,c_n)$ is true. As usual, ``{\bf false}''   means ``not true''.   The concepts of truth and falsity extend to $\top$, $\bot$ and all $\neg,\mlc,\mld$-combinations of closed atoms in the standard way: $\top^*$ is always true; $\bot^*$ is always false; $(\neg A)^*$ is true iff $A^*$ is false; $(A\mlc B)^*$ is true iff both $A^*$ and $B^*$ are true;   $(A\mld B)^*$ is true iff at least one of $A^*,B^*$ is true.   
When a cirquent is represented as $A(t)$, we usually write $A^*(t)$ instead of $\bigl(A(t)\bigr)^*$.

Remember from \cite{cl16}  that,  when  we say  ``{\bf won}'' without specifying a player (as in the following definition), it always means ``won by the machine''.  Similarly for ``{\bf lost}''.

\begin{definition} \label{semantics}
Every closed cirquent $C$ and interpretation $^*$ induces a unique game $C^{*}$, which we refer to as ``$C$ {\bf under} $^*$'', defined as follows. The 
set $\legal{C^{*} }{}$ of legal runs of such a game is nothing but {\em LegalRuns}. Since  $\legal{C^*}{}$ does not depend on $C$ or $^*$, subsequently we shall simply say ``legal run'' rather than  ``legal run of $C^{*}$''. 
The $\win{C^{*}}{}$ component of the game $C^{*}$ is defined by stipulating that a legal run  $\Gamma$ is a won   run of $C^{*}$ iff one of the following conditions is satisfied: 

1.  $C$  is a  literal and $C^*$ is true.

2. $C$ has the form $A_0\mlc A_1$ (resp. $A_0\mld A_1$) and, for both (resp. at least one) $i\in\{0,1\}$,  $\Gamma$ is a won run of $A_{i}^{*}$.

3. $C$ is a $\Gamma$-resolved $\adc$-, $\add$-, $\ada$- or $\ade$-rooted cirquent and,  where $B$ is the resolvent, $\Gamma$ is a won run of $B^{*}$. 

4. $C$ is a $\Gamma$-unresolved $\adc$- or $\ada$-rooted cirquent.

\end{definition}

When an interpretation $^*$ is fixed in the context or is irrelevant, by abuse of notation we may omit explicit references to it, and identify a cirquent $C$ with the game $C^*$. For instance, we may say that the machine wins $C$ instead of saying that the machine wins $C$ under $^*$.

\begin{definition} \label{krisha}
 Consider a closed cirquent $C$.

1. For an interpretation $^*$, a {\bf solution} of $C$ under $^*$, or simply a solution of $C^*$, is an HPM $\cal H$ such that ${\cal H}\models   C ^{*}$. 
 We say that $C$ is {\bf computable} under $^*$, or simply that $C^*$ is computable, iff $C^*$ has a solution. 

2. A {\bf logical}  {\bf solution} of $C$ is an HPM $\cal H$ such that, for any interpretation $^*$, $\cal H$ is a solution of $ C ^*$. We say that $C$ is (logically)  {\bf valid} if it has a logical solution; otherwise $C$ is  {\bf invalid}.
\end{definition}

\section{Axiomatics}

Just like $\lll$, our present system $\mmm$ has $\top$ as its only axiom.  The inference rules of  $\mmm$ are listed below,  where   all  notational conventions from \cite{cl16} remain  in force.  Note that all rules of $\lll$ are also rules of $\mmm$, even if ``somewhat'' renamed.

\begin{description}
\item[Por-commutativity:]  $X[B\mld A]\leadsto X[A\mld B]$.
\item[Pand-commutativity:] $X[B\mlc A]\leadsto X[A\mlc B]$. 
\item[Por-associativity:]  $X[A\mld(B\mld C)]\leadsto  X[(A\mld B)\mld C ]$.
\item[Pand-associativity:] $X[A\mlc(B\mlc C)]\leadsto  X[(A\mlc B)\mlc C ]$.
\item[Por-identity:]  $X[A]\leadsto X[A\mld \bot]$.
\item[Pand-Identity:]  $X[A]\leadsto  X[A\mlc \top]$.  
\item[Por-domination:]  $X[\top]\leadsto X[ A\mld \top]$.
\item[Pand-domination:]  $X[\bot]\leadsto  X[A\mlc \bot]$. 
\item[Left chor-choosing:]  $X[A_1,\cdots,A_n]\leadsto  X[A_1 \add^\mathfrak{c} B_1,\cdots,A_n \add^\mathfrak{c} B_n]$, \ where $A_1\add^\mathfrak{c} B_1$, $\cdots$, $A_n \add^\mathfrak{c} B_n$ are all $\add^\mathfrak{c}$-rooted subcirquents of the conclusion. 
\item[Right chor-choosing:] $X[B_1,\cdots,B_n]\leadsto  X[A_1 \add^\mathfrak{c} B_1,\cdots,A_n \add^\mathfrak{c} B_n]$,  \ where $A_1\add^\mathfrak{c} B_1$, $\cdots$, $A_n \add^\mathfrak{c} B_n$ are all $\add^\mathfrak{c}$-rooted subcirquents of the conclusion.
\item[Chexists-choosing:] \(X[A_1(a),\cdots,A_n(a)]\leadsto  X[\ade^\mathfrak{c} x_1A_1(x_1),\cdots,\ade^\mathfrak{c} x_nA_n(x_n)],\)  where $a$ is any constant and $\ade^\mathfrak{c} x_1A_1(x_1),\cdots,\ade^\mathfrak{c} x_nA_n(x_n)$ are all $\ade^\mathfrak{c}$-rooted subcirquents of the conclusion.    
\item[Left chand-cleansing:] $X\bigl[Y[A]\adc^\mathfrak{c} C\bigr]\leadsto X\bigl[Y[A\adc^\mathfrak{c} B]\adc^\mathfrak{c} C\bigr]$.
\item[Right chand-cleansing:]  $X\bigl[C \adc^\mathfrak{c} Y[B]\bigr]\leadsto X\bigl[C \adc^\mathfrak{c} Y[A\adc^\mathfrak{c} B]\bigr]$.
\item[Chall-cleansing:] $X\bigl[\ada^\mathfrak{c} xY[A(x)]\bigr]\leadsto X\bigl[\ada^\mathfrak{c} xY[\ada^\mathfrak{c} yA(y)]\bigr]$. 
\item[Pand-distribution:]     $X[(A\mld C)\mlc(B \mld C)]\leadsto  X[(A\mlc B)\mld C]$.
\item[Chand-distribution:]  $X[(A\mld C)\adc^\mathfrak{c} (B \mld C)]\leadsto  X[(A\adc^\mathfrak{c} B)\mld C]$.  
\item[Chall-distribution:] $X [\ada^\mathfrak{c}x (A \mld B ) ]\leadsto  X[\ada^\mathfrak{c} xA \mld B]$, where $x$ has no free occurrences in $B$.   
\item[Trivialization:] $X[\top]\leadsto  X[\neg A\mld A]$,  where $A$ is a nonlogical atom. 
\item[Chandchotomy:] \[X \bigl[\bigl( (A\mlc  C \adc^\mathfrak{b} D  ) \adc^\mathfrak{a}  ( B\mlc  C \adc^\mathfrak{b} D  )\bigr)\adc^\mathfrak{c}  \bigl(( A\adc^\mathfrak{a} B \mlc C )
      \adc^\mathfrak{b} (  A\adc^\mathfrak{a} B \mlc D ) \bigr)\bigr] \leadsto  X[ A\adc^\mathfrak{a} B \mlc  C \adc^\mathfrak{b} D ],\]  
where $\mathfrak{c}$ does not occur in the conclusion.
\item[Challchotomy:] \[X [\ada^\mathfrak{a}x  (A \mlc \ada^\mathfrak{b}yB   )\adc^\mathfrak{c} \ada^\mathfrak{b}y  (\ada^\mathfrak{a} xA \mlc B   ) ] \leadsto  X [\ada^\mathfrak{a}xA \mlc \ada^\mathfrak{b}yB  ],\]  where $\mathfrak{c}$ does not occur in the conclusion, $x$ has no free occurrences in $\ada^\mathfrak{b}yB $ and $y$ has no free occurrences in $\ada^\mathfrak{a}xA $.
\item[Chandallchotomy:] \[X \bigl[\bigl( (A\mlc \ada^\mathfrak{b}xC   )\adc^\mathfrak{a} (B\mlc \ada^\mathfrak{b}xC  )\bigr)\adc^\mathfrak{c} \ada^\mathfrak{b} x  ( A\adc^\mathfrak{a} B \mlc C   )\bigr] \leadsto  X [ A\adc^\mathfrak{a} B \mlc \ada^\mathfrak{b}xC  ],\]  where $\mathfrak{c}$ does not occur in the conclusion and $x$ has no free occurrences in $A\adc^\mathfrak{a} B$. 
\item[Chand-splitting:] $A,B\leadsto  A\adc^\mathfrak{c} B$,  where neither $A$ nor $B$ has   occurrences of $\mathfrak{c}$. 
\item[Chall-splitting:] $A(a)\leadsto  \ada^\mathfrak{c} xA(x)$,  where neither $\mathfrak{c}$ nor    $a$ has occurrences in $A(x)$.  
\end{description}

We will be using the word ``{\bf Commutativity}'' as a common name of Por-commutativity and Pand-commutativity.  Similarly for all other rules. Throughout the rest of this article, ``$A$ is provable'' always means ``$A$ is provable in $\mmm$'', written $\mmm\vdash A$.  

\begin{example}\label{exm1} The cirquent $\ade^\mathfrak{a}xp(x)\mld \ade^\mathfrak{b}xp(x)\mli \ade^\mathfrak{c}xp(x)$ can be shown to be unprovable. However, it becomes provable if the two clusters $\mathfrak{a}$ and $\mathfrak{b}$ are the same. Below is a proof of  $\ade^\mathfrak{a}xp(x)\mld \ade^\mathfrak{a}xp(x)\mli \ade^\mathfrak{c}xp(x)$, i.e., of $\bigl(\ada^\mathfrak{a}x\neg p(x)\mlc \ada^\mathfrak{a}x\neg p(x)\bigr)\mld \ade^\mathfrak{c}xp(x)$.

1. $\top$ \hspace{5pt} Axiom 

2. $\top \mlc \top$ \hspace{5pt} Pand-identity: 1

3. $\bigl(\neg p(a)\mld p(a)\bigr)\mlc \bigl(\neg p(a) \mld p(a)\bigr)$ \hspace{5pt} Trivialization: 2 (twice)

4. $\bigl(\neg p(a)\mlc \neg p(a)\bigr)\mld p(a)$\hspace{5pt} Pand-distribution: 3

5. $\bigl(\neg p(a)\mlc \neg p(a)\bigr)\mld \ade^\mathfrak{c}xp(x)$ \hspace{5pt} Chexists-choosing: 4

6. $\ada^\mathfrak{a}x \Bigl(\bigl(\neg p(x)\mlc \neg p(x)\bigr)\mld \ade^\mathfrak{c}xp(x)\Bigr)$ \hspace{5pt} Chall-splitting: 5

7. $\ada^\mathfrak{a}x \bigl(\neg p(x)\mlc \neg p(x)\bigr)\mld \ade^\mathfrak{c}xp(x)$ \hspace{5pt} Chall-distribution: 6

8. $\Bigl(\ada^\mathfrak{a}x \bigl(\neg p(x)\mlc \neg p(x)\bigr)\mld \ade^\mathfrak{c}xp(x)\Bigr)\adc^\mathfrak{b}\Bigl(\ada^\mathfrak{a}x \bigl(\neg p(x)\mlc \neg p(x)\bigr) \mld \ade^\mathfrak{c}xp(x)\Bigr)$ \hspace{5pt} Chand-splitting: 7,7

9. $\ada^\mathfrak{a}x \bigl(\neg p(x)\mlc \neg p(x)\bigr)\adc^\mathfrak{b}\ada^\mathfrak{a}x \bigl(\neg p(x)\mlc \neg p(x)\bigr) \mld \ade^\mathfrak{c}xp(x)$ \hspace{5pt}Chand-distribution: 8
 
10. $\ada^\mathfrak{a}x \bigl(\neg p(x)\mlc \ada^\mathfrak{a}x\neg p(x)\bigr)\adc^\mathfrak{b}\ada^\mathfrak{a}x \bigl(\ada^\mathfrak{a}x \neg p(x)\mlc \neg p(x)\bigr) \mld \ade^\mathfrak{c}xp(x)$ \hspace{5pt} Chall-cleansing: 9 (twice)

11. $\bigl(\ada^\mathfrak{a}x\neg p(x)\mlc \ada^\mathfrak{a}x\neg p(x)\bigr)\mld \ade^\mathfrak{c}xp(x)$ \hspace{5pt} Challchotomy: 10
\end{example}

\section{The preservation lemma}

A {\bf surface occurrence} of a subcirquent or an operator  in a given cirquent is an occurrence which is not in the scope of a choice operator (i.e., of $\adc^\mathfrak{c}$, $\add^\mathfrak{c}$, $\ada^\mathfrak{c}$ or $\ade^\mathfrak{c}$ for whatever $\mathfrak{c}$).

\begin{definition}\label{residue}
Given a closed cirquent $C$ and a legal run $\Gamma$, the {\bf $\Gamma$-residue} of $C$ is the  $\adc,\add,\ada,\ade$-free cirquent obtained from $C$ as a result of repeatedly replacing  until no longer possible: 

1. every surface occurrence of every $\Gamma$-resolved $\adc$-, $\add$-, $\ada$- or $\ade$-rooted subcirquent by its resolvent;

2. every surface occurrence of every $\Gamma$-unresolved $\adc$- or $\ada$-rooted subcirquent by $\top$;

3. every surface occurrence of every $\Gamma$-unresolved $\add$- or $\ade$-rooted subcirquent by $\bot$.

\end{definition}

\begin{lemma}\label{f1}
Consider any closed cirquent $C$, interpretation $^*$  and legal run $\Gamma$. Let $R_C$ be the $\Gamma$-residue of $C$. Then  $\Gamma$ is a won (by the machine) run of $C^*$ iff $R_{C}^{*}$ is true. 
\end{lemma}
\begin{proof} Let $C$, $^*$, $\Gamma$, $R_{C}$ be as above. We proceed by induction on the complexity of $C$. 

If $C$ is a literal, then $R_C=C$. Hence, from clause 1 of Definition \ref{semantics} (and the fact that all other clauses of that definition are about non-literal cases), $\Gamma$ is a won run of $C^*$  iff $R_{C}^{*}$ is true.

If $C$ is of the form $A\mlc B$, then obviously $R_{C} =R_A\mlc R_B$, where $R_A$ is the residue of $A$ and $R_B$ is the residue of $B$. By Definition \ref{semantics},  $\Gamma$ is a won run of  $(A\mlc B)^*$ iff it is a won run of both $ {A}^{*}$ and $ {B}^{*}$. But, by the induction hypothesis, $\Gamma$ is a won run of both $A^*$ and $B^*$   iff both $R_{A}^{*}$ and $R_{B}^{*}$ are true. ``Both $R_{A}^{*}$ and $R_{B}^{*}$ are true'', in turn, means nothing but that $R_{A}^{*}\mlc R_{B}^{*}$, i.e., $R_{C}^{*}$, is true.    

The case of $C$ having the form $A\mld B$ will be handled similarly.

Suppose $C$ is of the form $A\adc^\mathfrak{c} B$. If $C$ is unresolved, then, by clause 4 of Definition \ref{semantics}, $\Gamma$ is a won run of $C^*$. But then, as desired, $R_{C}^{*}$ is true because $R_{C}=\top$. Now suppose $C$ is resolved. Without loss of generality we may assume that the resolvent is $A$. Notice that then $R_{C}$ is the residue of not only $C$ but also of $A$. We have: $\Gamma$ is a won run of $C^*$ iff  (by clause 3 of Definition \ref{semantics}) it is a won run of $A^*$ iff (by the induction hypothesis) $R_{C}^{*}$ is true, as desired.  

The remaining cases of $C$ having the form $A\add^\mathfrak{c} B$, $\ada^\mathfrak{c} x A(x)$ or  $\ade^\mathfrak{c} x A(x)$ will be handled similarly. 
\end{proof}

\begin{lemma} \label{pres}
\ 

1. Each application of any of the rules of $\mmm$ preserves logical validity in the premises-to-conclusion direction, i.e., if all premises are   valid, then so is the conclusion.   
 
2. Each application of any of the rules of $\mmm$ other than (the three versions of) Choosing also preserves logical validity in the conclusion-to-premises direction, i.e., if the conclusion is   valid, then so are all premises. 
\end{lemma}

\begin{proof} As in \cite{cl16}, we will implicitly rely on the clean environment assumption, allowing us to rule out the possibility that the environment ever makes any illegal moves. We shall also implicitly rely on the straightforward fact that    if, for every   interpretation $^*$,  computability of a cirquent $A$ under $^*$ implies computability   of a cirquent $B$ under the same $^*$,    then  logical validity  of $A$ implies logical   validity of $B$.   

If $E\leadsto F$ is an application of any of the rules other than Splitting, Chotomy or Choosing, it is not hard to see that, for any interpretation $^*$,  $E^*$  and $F^*$ are identical as games. So, a logical solution of $E$ is automatically a logical solution of $F$, and vice versa. Let us just look at Chall-cleansing as an   example. Consider an application $E\leadsto F$ of this rule, where $E=X\bigl[\ada^\mathfrak{c} xY(x)[A(x)]\bigr]$ and $F=X\bigl[\ada^\mathfrak{c} xY(x)[\ada^\mathfrak{c} yA(y)]\bigr]$. Fix some arbitrary interpretation for the present context, and let $\Gamma$ be an arbitrary legal run.  We want to show that $\Gamma$ is a won run of $E $ iff it is a won run of $F $. If $\mathfrak{c}$ is $\Gamma$-unresolved, then  $\Gamma$ is a won run of the  $ \ada^\mathfrak{c} xY(x)[\ada^\mathfrak{c} yA(y)] $ component of the conclusion  just as it is a won run of  the $ \ada^\mathfrak{c} xY(x)[A(x)] $ component of the premise. Then, since $E$ and $F$ only differ in that one has $\ada^\mathfrak{c} xY(x)[\ada^\mathfrak{c} yA(y)]$ where the other has $\ada^\mathfrak{c} xY(x)[A(x)]$, we find that $\Gamma$ is a won run of both games $E ,F $ or neither. Now assume $\mathfrak{c}$ is resolved, i.e., $\Gamma$ contains the move $\mathfrak{c}.a$ for some constant $a$. Then  $\Gamma$ is a won run of  $ \ada^\mathfrak{c} xY(x)[\ada^\mathfrak{c} yA(y)] $ iff it is a won run  of $ Y(a)[A(a)] $  iff it is a won run of   $ \ada^\mathfrak{c} xY(x)[A(x)] $. This, again,    implies that $\Gamma$ is a won run of both games $ E , F $ or neither.  

{\em Chor-choosing} is taken care of in the proof of Lemma 6.1 of \cite{cl16}, which shows that this rule preserves computability under any given interpretation.  The same can be said about either direction of   
{\em Chand-splitting}.  

Consider an application $X[A_1(a),\cdots,A_n(a)]\leadsto  X[\ade^\mathfrak{c} x_1A_1(x_1),\cdots,\ade^\mathfrak{c} x_nA_n(x_n)]$ of {\em Chexists-choosing}, and assume $\cal M$ is a logical 
solution of the premise. Let ${\cal N}$ be an HPM that, at the beginning of the play, makes the move $\mathfrak{c}.a$, after which it plays exactly as 
$\cal M$ would. Obviously $\cal N$ 
is a logical solution of  the conclusion.  

Of the three {\em Chotomy} rules, let us just consider Challchotomy, with the remaining two rules being similar. Consider an application 
\[X [\ada^\mathfrak{a}x  (A \mlc \ada^\mathfrak{b}yB   )\adc^\mathfrak{c} \ada^\mathfrak{b}y  (\ada^\mathfrak{a} xA \mlc B   ) ] \leadsto  X [\ada^\mathfrak{a}xA \mlc \ada^\mathfrak{b}yB  ] \] 
of Challchotomy.  It is not hard to see that, under any interpretation,  any won run of the conclusion  is also (``even more so'') a won run of the premise, meaning that a   solution of the conclusion  would automatically also be a   solution of the premise. This takes care of the conclusion-to-premise direction. For the premise-to-conclusion direction, fix an arbitrary interpretation and assume that $\cal M$ is a solution of the premise. Let $\cal N$ be an HPM that, until  it sees that its environment has resolved either $\mathfrak{a}$ or $\mathfrak{b}$, plays just as $\cal M$ would play in the scenario where $\mathfrak{c}$ is not (yet) resolved but otherwise $\cal M$'s imaginary environment is making the same moves as $\cal N$'s real environment is making. If and when it sees that a move $\mathfrak{a}.i$ (resp. $\mathfrak{b}.i$) has been made by its environment, $\cal N$  imagines that $\cal M$'s environment has correspondingly made not only the same move $\mathfrak{a}.i$ (resp. $\mathfrak{b}.i$) but also  $\mathfrak{c}.0$  (resp. $\mathfrak{c}.1$), and  continues playing exactly as $\cal M$ would continue playing in that case. With a little thought, $\cal N$ can be seen to be a solution of the conclusion.

Finally, consider an application $A(a)\leadsto  \ada^\mathfrak{c} xA(x)$ of {\em Chall-splitting}. 

For the conclusion-to-premise direction, assume $\cal M$ is a logical solution of $ \ada^\mathfrak{c} xA(x)$. Let $\cal N$ be the HPM that plays exactly as $\cal M$ would play in the scenario where, at the very beginning of the play, the environment makes the move $\mathfrak{c}.a$. It is not hard to see that $\cal N$ is a logical solution of $A(a)$.  

For the premise-to-conclusion direction,
assume $\cal M$ is a logical solution of $A(a)$. 
Let $\cal N$ be an HPM that plays $\ada^\mathfrak{c} xA(x)$ as follows. While continuously polling its run tape, $\cal N$ maintains a list $L$ of  moves, initially empty. This is just to keep track of which moves made by the environment have already been ``processed'' by $\cal N$. Call the moves that are not in $L$ {\em unprocessed}. $\cal N$ further maintains a partial function   $f:\mathbb{N}\rightarrow \mathbb{N}$. Call the value of $f(c)$ the {\em image} of $c$, and call the constants at which $f$ is not (yet) defined {\em imageless}. Initially, the image of each constant occurring in $A(x)$ is that constant itself, and all other constants are imageless. 

At the beginning of the play, $\cal N$ waits till the environment makes the move $\mathfrak{c}.b$ for some constant $b$. 
If and when this happens (and if not, $\cal N$ wins $\ada^\mathfrak{c} xA(x)$ by doing nothing), $\cal N$ adds $\mathfrak{c}.b$ to $L$, declares $b$ to be the image of (the so far imageless) $a$, and  starts simulating an imaginary play of $A(a)$ by $\cal M$. In this simulation:
\begin{itemize}
  \item  Whenever $\cal N$ sees an unprocessed move $\mathfrak{a}.i$ on its run tape where $\mathfrak{a}$ is a $\adc$-cluster, it adds this move to $L$  and appends $\bot \mathfrak{a}.i$  to the content of the imaginary run tape of $\cal M$.
  \item  Whenever $\cal N$ sees an unprocessed move $\mathfrak{a}.c$ on its run tape where $\mathfrak{a}$ is a $\ada$-cluster, it adds this move to $L$    and appends  $\bot \mathfrak{a}.d$  to the content of the imaginary run tape of $\cal M$, where $d$ is an (say, the smallest) imageless constant; after that $\cal M$ declares $c$ to be the image of $d$.  
\item Whenever $\cal N$ sees that the simulated $\cal M$ made a move $\mathfrak{a}.i$ 
where $\mathfrak{a}$ is a $\add$-cluster,  it makes the same move  $\mathfrak{a}.i$ in its real play. 
\item Whenever $\cal N$ sees that the simulated $\cal M$ made a move $\mathfrak{a}.c$ where $\mathfrak{a}$ is a $\ade$-cluster and $c$ is imageless,  $\cal N$  makes the  same move  $\mathfrak{a}.c$ in its real play,  and declares $c$ to be its own image.  
\item Whenever $\cal N$ sees that the simulated $\cal M$ made a move $\mathfrak{a}.c$ where $\mathfrak{a}$ is a $\ade$-cluster and $c$ is not imageless,  $\cal N$  makes the   move  $\mathfrak{a}.d$ in its real play, where $d$ is the image of $c$.  
\end{itemize} 
We claim that $\cal N$ is a logical solution of  $\ada^\mathfrak{c} xA(x)$. To see this, consider an arbitrary interpretation $^*$ and an arbitrary ``real'' play by $\cal N$. Let $\Gamma_{\cal N}$ be the run that has taken place in that play, and $\Gamma_{\cal M}$ be the run that has correspondingly taken place in $\cal M$'s  play as imagined by $\cal N$. Our goal is to show that $\Gamma_{\cal N}$ is a won run of $\bigl(\ada^\mathfrak{c} xA(x)\bigr)^*$. Let $R_{\cal M}$ be the $\Gamma_{\cal M}$-residue of $A(a)$, and let $c_1,\cdots,c_n$ be all constants occurring in $R_{\cal M}$. Notice that each such constant has acquired  an image at some time during the work of $\cal N$ (and never lost or changed it afterwards).   Let $R_{\cal N}$ be the result of replacing in $R_{\cal M}$ each occurrence of each $c_i\in\{c_1,\cdots,c_n\} $ by the image of $c_i$. Let $^\circ$ be an interpretation such that, for any atom $X$ of $R_{\cal M}$, $X^\circ$ is true iff so is $Y^*$, where $Y$ is the result of replacing all constants in $X$ by their images. 

An analysis of the work of $\cal N$, details of which are left to the reader, reveals  that   $R_{\cal N}$ is the $\Gamma_{\cal N}$-residue of $\ada^\mathfrak{c} x A(x)$, and that $R_{\cal N}^{*}$ is true iff  so is $R_{\cal M}^{\circ}$. Since $\cal M$ is a logical solution of $A(a)$, $\Gamma_{\cal M}$ is a won run of $A^\circ(a)$. By Lemma \ref{f1}, this implies that $R_{\cal M}^{\circ}$ is true. Hence so is  $R_{\cal N}^{*}$, which, again by Lemma \ref{f1}, implies that $\Gamma$ is a won run of $\bigl(\ada^\mathfrak{c} xA(x)\bigr)^*$, as desired.  
\end{proof}

\section{Rank  and purification} 
Our proofs in this section closely  follow those given in Section 7 of \cite{cl16}.
As in \cite{cl16}, $^na$ means ``tower of $a$'s of height $n$'' (tetration), defined inductively by $^1 a= a$ and $^{n+1}a=a^{( ^na)}$.   
\begin{definition}\label{rankdef}
The {\bf rank} $\Rank(C)$ of a cirquent $C$ is the number defined as follows:  

1. If $C$ is  a literal, then its rank is $1$.

2. If $C$ is  $A\add^\mathfrak{c} B$ or $A\adc^\mathfrak{c} B$, then its rank is $\Rank(A)+\Rank(B)$.

3. If $C$ is  $\ade^\mathfrak{c} xA$ or $\ada^\mathfrak{c}x A$, then its rank is $\Rank(A)+1$.

4. If $C$ is  $A\mlc B$, then its rank is $5^k$, where $k=\Rank(A)+\Rank(B)$.\footnote{In fact, a smaller number can be taken here and below instead of $5$, but why bother.}

5. If $C$ is $A\mld B$, then its rank is \ $^k5$, where $k=\Rank(A)+\Rank(B)$.
\end{definition}

Due to due to the monotonicity of the functions $x+y$, $x+1$, $5^x$ and $^x5$, we have:
\begin{lemma}\label{monot}
The rank function is {\bf monotone} in the following sense. Consider a cirquent $A$ with a subcirquent $B$. Assume $B'$ is a cirquent with {\em $\Rank(B')< \Rank(B)$}, and $A'$ is the result of replacing an occurrence of $B$ by $B'$ in $A$. Then {\em $\Rank(A')< \Rank(A)$}.
\end{lemma}

\begin{definition}\label{pd}
 We say that a cirquent $E$ is {\bf pure} iff the following conditions are satisfied: 

1. $E$ has no surface occurrence  of $\bot$ unless $E$ itself is $\bot$. 

2. $E$ has no surface occurrence of $\mlc$ which is in the scope of $\mld$. 

3. $E$ has no surface occurrence of $\adc^\mathfrak{c}$ or $\ada^\mathfrak{c}$ (whatever cluster $\mathfrak{c}$) which is in the scope of $\mld$.
 

4. $E$ has no surface occurrence of the form $A_1\mld\cdots\mld A_n$ such that, for some atom $A$, both 
$A$ and $\neg A$ are among $A_1,\cdots,A_n$. 

5. $E$ has no surface occurrences of $\top$ unless $E$ itself is $\top$. 

6. If $E$ is of the form $A_1\mlc\cdots\mlc A_n$ ($n\geq 2$), then at least one $A_i$ ($1\leq i\leq n$) is neither $\adc$- nor $\ada$-rooted.  


7. If $E$ is of the form $A\adc^\mathfrak{c} B$, then neither $A$ nor $B$ contains the cluster $\mathfrak{c}$. 

8. If $E$ is of the form $\ada^\mathfrak{c} xA$, then   $A$ does not contain  the cluster $\mathfrak{c}$. 
\end{definition}

Below we describe a procedure which takes a cirquent $E$ and applies to it a series of modifications. Each modification changes the value of $E$ so that the old value of $E$ follows from the new value by one of the  rules  of $\mmm$ other than Choosing and Splitting. The procedure is divided into eight stages, and the purpose of each stage $i\in\{1,\cdots,8\}$ is to make $E$ satisfy the corresponding condition $\#i$ of Definition \ref{pd}.\vspace{5pt}   

{\bf Procedure Purification} applied to a cirquent $E$:  
Starting from Stage 1, each  of the following eight stages is a loop that should be iterated until it no longer modifies (the current value of) $E$; then the procedure goes to the next stage, unless the current stage was Stage 8, in which case the procedure terminates and returns (the then-current value of) $E$.  

{\em Stage 1}:   If $E$ has a surface occurrence of the form $\bot\mld A$ or  $A\mld\bot$, change the latter to $A$  
using Por-identity perhaps in combination with Por-commutativity.  Next, if $E$ has a surface occurrence of the form $\bot\mlc A$  or $A\mlc\bot$, change it to $\bot$  using Pand-domination   perhaps in combination with Pand-commutativity.

{\em Stage 2}: If $E$ has a surface occurrence of the form $(A\mlc B)\mld C$ or $C\mld(A\mlc B)$, change it to $(A\mld C)\mlc (B\mld C)$ using Pand-distribution perhaps in combination with Por-commutativity. 

{\em Stage 3}:   (a) If $E$ has a surface occurrence of the form  $A\adc^\mathfrak{c} B \mld C$ or $C\mld A\adc^\mathfrak{c} B $, change it to $(A\mld C)\adc^\mathfrak{c} (B \mld C)$ using Chand-distribution perhaps in combination with Por-commutativity. (b) Next, if $E$ has a surface occurrence of the form  $\ada^\mathfrak{c}xA  \mld B$ or $B\mld\ada^\mathfrak{c}xA$, change it to $\ada^\mathfrak{c}x (A \mld B )$ using Chall-distribution perhaps in combination with Por-commutativity.

{\em Stage 4}: If $E$ has a surface occurrence of the form $A_1\mld \cdots\mld A_n$ and, for some atom $A$, both 
$A$ and $\neg A$ are among $A_1,\cdots,A_n$, change $A_1\mld \cdots\mld A_n$ to $\top$ using Trivialization, perhaps in combination with Por-domination, Por-commutativity and Por-associativity.

{\em Stage 5}: If $E$ has a surface occurrence of the form $\top\mld A$ or $A\mld \top$, change it to  $\top$  using Por-domination   perhaps in combination with Por-commutativity. Next, if $E$ has a surface occurrence of the form
$\top \mlc  A$ or $A\mlc\top$,  change it to $A$  
using Pand-identity perhaps in combination with Pand-commutativity.  

{\em Stage 6}: In all three cases below, $\mathfrak{c}$ is a $\adc$-cluster not occurring in $E$. (a) If $E$ has a surface occurrence of the form $ A \adc^\mathfrak{a} B \mlc  C \adc^\mathfrak{b} D $, change it to 
$\bigl( (A\mlc  C \adc^\mathfrak{b} D  )\adc^\mathfrak{a}  (B\mlc  C \adc^\mathfrak{b} D  )\bigr) 
\adc^\mathfrak{c} \bigl( ( A\adc^\mathfrak{a} B \mlc C )\adc^\mathfrak{b} ( A\adc^\mathfrak{a} B \mlc D )\bigr)$  using Chandchotomy. 
  (b) Next, if $E$ has a surface occurrence of the form $\ada^\mathfrak{a}xA \mlc \ada^\mathfrak{b}yB $, change it to 
$\ada^\mathfrak{a}x (A \mlc \ada^\mathfrak{b}yB  )\adc^\mathfrak{c} \ada^\mathfrak{b}y  (\ada^\mathfrak{a} xA \mlc B   )$ using Challchotomy.  (c) Next, if $E$ has a surface occurrence of the form $ A\adc^\mathfrak{a} B \mlc \ada^\mathfrak{b}xC $ or 
$ \ada^\mathfrak{b}xC\mlc A\adc^\mathfrak{a} B $,      
  change it to $\bigl( (A\mlc \ada^\mathfrak{b}xC   )\adc^\mathfrak{a} (B\mlc \ada^\mathfrak{b}xC  )\bigr)\adc^\mathfrak{c} \ada^\mathfrak{b} x  ( A\adc^\mathfrak{a} B \mlc C   )$ using Chandallchotomy perhaps in combination with Pand-commutativity.

{\em Stage 7}: If $E$ is  of the form $X[A\adc^\mathfrak{c} B]\adc^\mathfrak{c} C$ (resp. $C\adc^\mathfrak{c} X[A\adc^\mathfrak{c} B]$), change it to $X[A]\adc^\mathfrak{c}C$ (resp. $C\adc^\mathfrak{c} X[B]$) using Left (resp. Right) chand-cleansing. 

{\em Stage 8}: If $E$ is of the form $ \ada^\mathfrak{c} xX[\ada^\mathfrak{c} yA(y)] $, change it to $ \ada^\mathfrak{c} xX[A(x)] $  using Chall-cleansing. 

\

\begin{lemma}\label{terminate}  
Each stage of the Purification procedure strictly reduces the rank of $E$.  
\end{lemma}

\begin{proof} Each stage replaces an occurrence of a subcirquent $A$ of $E$ by some cirquent $B$. In view of Lemma \ref{monot}, in order to show that such a replacement reduces the rank of $E$, it is sufficient to show that $\Rank(B)<\Rank(A)$.   Here we shall only consider Stages 3(b), 6(b,c) and 8, as all other stages or cases are covered in the proof of Lemma 7.4 of \cite{cl16}.

{\em Stage 3(b)}: With $k$ abbreviating $\Rank(A)+\Rank(B)$, the rank of $\ada^\mathfrak{c}xA  \mld B$ or $B\mld\ada^\mathfrak{c}xA$ is $^{k+1}5$, and the rank of   $\ada^\mathfrak{c}x (A \mld B )$ is $^{k }5+1$. The latter is clearly smaller than the former. 

{\em Stage 6(b)}:   With $k$ abbreviating $\Rank(A)+\Rank(B)$, the rank of $\ada^\mathfrak{a}xA \mlc \ada^\mathfrak{b}yB $   is $5^{k +2}$, and the rank of    
$\ada^\mathfrak{a}x (A \mlc \ada^\mathfrak{b}yB  )\adc^\mathfrak{c} \ada^\mathfrak{b}y  (\ada^\mathfrak{a} xA \mlc B   )$
is $5^{k+1 }+5^{ k+1 }   $. Of course the latter is smaller than the former.

{\em Stage 6(c)}:  With $a$, $b$ and $c$ standing for the ranks of $A$, $B$ and $C$, respectively, the rank of $ A \adc^\mathfrak{a} B \mlc \ada^\mathfrak{b}xC$ or  $\ada^\mathfrak{b}xC\mlc  A \adc^\mathfrak{a} B $ is $5^{a+b+c+1}$, and the rank of    
$\bigl( (A\mlc \ada^\mathfrak{b}xC   )\adc^\mathfrak{a}  (B\mlc \ada^\mathfrak{b}xC )\bigr)\adc^\mathfrak{c} \ada^\mathfrak{b}x ( A \adc^\mathfrak{a} B \mlc C  ) $
is $5^{ a+c+1 }+5^{b+c+1 }+5^{ a+b+c }  +1$. Obviously the latter is smaller than the former.

{\em Stage 8}: Each iteration of this stage replaces a subcirquent $\ada^\mathfrak{c} yA(y)$ by $A(x)$. Since the rank of  $ A(x)$ is obviously the same as the rank of  $A(y)$, the rank $\Rank\bigl(A(x)\bigr)+1$ of $\ada^\mathfrak{c} xA(x)$ is greater than the rank $\Rank\bigl(A(y)\bigr)$ of $A(y)$.   
 \end{proof}

Where $A$ is the initial value of $E$ in the Purification procedure and $B$ is its final value (which exists by Lemma \ref{terminate}), we call $B$ the {\bf purification} of $A$.

\begin{lemma}\label{pl}
 For any closed cirquent $E$ and its purification $F$, we have: 

1. If $F$ is provable in $\mmm$, then so is $E$. 

2. $E$ is   valid iff so is $F$. 

3. $F$ is pure. 

4. The rank of $F$ does not exceed the rank of $E$.  
\end{lemma}

\begin{proof} {\em Clause 1}: When obtaining $F$ from $E$, each transformation performed during the Purification procedure applies, in the conclusion-to-premise direction, one of the inference rules of $\mmm$. Reversing the order of those transformations, we get a derivation of $E$ from $F$ in $\mmm$. Appending that derivation to a proof of $F$ 
(if one exists) yields a proof of $E$. 

{\em Clause 2}: Immediate from the two clauses of Lemma \ref{pres} and the fact that, when obtaining $F$ from $E$ 
using the Purification procedure, Choosing is never used. 

{\em Clause 3}: One by one, Stage 1 eliminates all surface occurrences of $\bot$ in $E$ (unless $E$ itself is $\bot$). So, at the end of the stage, $E$ satisfies condition 1 of Definition \ref{pd}. None of the subsequent steps make $E$ violate that condition, so $F$, too, satisfies that condition.
 Similarly, a routine examination of the situation reveals that Stage 2 (resp. 3, \ldots, resp. 8) of the Purification procedure makes $E$ satisfy condition 2 (resp. 3, \ldots, resp. 8) of Definition \ref{pd}, and $E$ continues to satisfy that condition throughout the rest of the stages. So, $F$ is pure.   

{\em Case 4}:  Immediate from Lemma  \ref{terminate}. \end{proof}

\section{The soundness and completeness of $\mmm$}
\begin{theorem}\label{theo}
A closed cirquent $E$ is   valid if (soundness) and only if (completeness) $\mmm\vdash E$.
\end{theorem}

\begin{proof} The soundness part is immediate from clause 1 of Lemma \ref{pres} and the fact that the axiom $\top$ is   valid. The rest of this section is devoted to a proof of the completeness part. Pick an arbitrary closed cirquent $E$, and let $F$ be its purification.   We proceed by induction on the 
rank of $E$.  

By clauses 3-4 of Lemma \ref{pl}, $F$ is a pure  cirquent whose rank does not exceed that of 
$E$. We shall implicitly rely on this fact below. In view of $F$'s being pure, after some analysis it is clear that   one of the following conditions should be satisfied:

\begin{description}
  \item[Condition 0.] $F$ is either $\bot$ or a nonlogical literal.\vspace{-7pt}  
  \item[Condition 1.] $F$ is $\top$.\vspace{-7pt}
  \item[Condition 2.] $F$ is of the form $A_0 \add^\mathfrak{c}  A_1$.\vspace{-7pt}
  \item[Condition 3.] $F$ is of the form $\ade^\mathfrak{c} xA(x)$.\vspace{-7pt}
  \item[Condition 4.] $F$ is of the form $A_0 \adc^\mathfrak{c} A_1$, and neither $A_0$ nor $A_1$ contains the cluster $\mathfrak{c}$.\vspace{-7pt} 
  \item[Condition 5.] $F$ is of the form $\ada^\mathfrak{c} xA(x)$, and $A(x)$ does not contain the cluster $\mathfrak{c}$.\vspace{-7pt}
  \item[Condition 6.] $F$ is of the form $A_1\mld\cdots\mld A_n$ ($n\geq 2$),  where each disjunct  is either    a nonlogical literal, or $\add$-rooted, or $\ade$-rooted; besides, for no atom $A$ do we have that both $A$ and $\neg A$ are among $A_1,\cdots,A_n$.\vspace{-7pt}
  \item[Condition 7.] $F$ is of the form $B_1\mlc\cdots\mlc B_m$ ($m\geq 2$), where at least one conjunct $B_e$ is  either   (0) a nonlogical literal, or (1) $\add$-rooted, or  (2) $\ade$-rooted, or (3) satisfies   Condition 6 in the role of $F$. 
\end{description}  

Assume $E$ is valid. We want to show that then $E$ is provable. For this, in view clause 1 of Lemma \ref{pl}, it is sufficient to show that $F$ is provable. Keep in mind that, by clause 2  of Lemma \ref{pl},  $F$ is   valid.   This immediately rules out Condition 0, because, of course, neither $\bot$ nor nonlogical literals are  valid. So, we only need to show that $F$ is provable in each of the following seven cases:

{\em Case 1}:  $F$ is $\top$  as in Condition 1. Then $F$ is an axiom and hence provable. 

{\em Case 2}: $F$ is $A_0 \add^\mathfrak{c} A_1$ as in Condition 2. Let $\cal M$ be a logical solution of $F$.  
Consider the work of $\cal M$ in the scenario where the environment does not move until $\cal M$ makes the move $\mathfrak{c}.i$ for one of $i\in\{0,1\}$. Sooner or later $\cal M$ has to make such a move, for otherwise $F$ would be lost 
due to being $\add^\mathfrak{c}$-rooted.  Since in the games that we deal with the order of moves is irrelevant, without loss of  generality we may assume that the move $\mathfrak{c}.i$ is made by $\cal M$ before any other moves. Let $D$ be the result of replacing in $F$ all subcirquents of the form $X_0\add^\mathfrak{c} X_1$ by $X_i$.  Observe that, after the move $\mathfrak{c}.i$ is made, in any scenario that may follow, $\cal M$ has to continue and win $D$. This means that  $\cal M$ is a logical solution of (not only $F$ but also) $D$. Thus, $D$ is valid. The rank of $D$ is of course smaller than that of $F$. Hence, by the induction hypothesis, $D $  is provable. Then so is $F$ because it follows from $D $ by Chor-choosing. 

{\em Case 3}:  $F$ is $\ade^\mathfrak{c} xA(x)$ as in Condition 3. This case is rather similar to the preceding one. Let  $\cal M$ be a logical solution of $F$. Consider the work of $\cal M$ in the scenario where the environment does not move until $\cal M$ makes the move $\mathfrak{c}.a$ for some constant $a$.
 Sooner or later $\cal M$ has to make such a move, for otherwise $ F$ would be lost  due to   being $\ade^\mathfrak{c}$-rooted. As in case 2, we may assume that the move  $\mathfrak{c}.a$  is made  before any other  moves. Let $D$ be the result of replacing in $F $ all subcirquents of the form $\ade^\mathfrak{c} yX(y)$ by $X(a)$.  After the move $\mathfrak{c}.a$ is made, in any scenario that may follow, $\cal M$ has to continue and win $D$. This means that  $\cal M$ is a logical solution of  $D$. The rank of  $D$ is smaller than that of $F$. Hence, by the induction hypothesis, $D$  is provable. Then so is $F$, as it follows from $D$ by Chexists-choosing. 

{\em Case 4}: $F=A_0 \adc^\mathfrak{c} A_1$ is as in Condition 4. By clause 2 of Lemma \ref{pres}, both $A_0$ and $A_1$ are valid, because $F$ follows from either one by Chand-splitting. Both $\Rank(A_0)$ and $\Rank(A_1)$ are smaller  than  $\Rank(F)$. Hence, by the induction hypothesis, both $A_0$ and $A_1$ are provable. Therefore, by Chand-splitting, so is $F$.

{\em Case 5}: $F=\ada^\mathfrak{c} xA(x)$ is as in condition 5.   Let $c$ be a constant not occurring in $F$. By clause 2 of Lemma \ref{pres},  $A(c)$ is   valid, because $F$ follows from it by Chall-splitting. The rank of $A(c)$ is smaller  than that of $F$. Hence, by the induction hypothesis,   $A(c)$ is provable. Therefore, by Chall-splitting, so is $F$.

{\em Case 6}:  $F=A_1\mld\cdots\mld A_n$ is as in Condition 6.    Not all of the cirquents $A_1,\ldots,A_n$ can be literals, for otherwise $F$ would be automatically lost under an interpretation that makes   all those literals false, contrary to our assumption that $F$  is   valid. With this observation in mind, without loss of generality, we may assume that, for some 
$k,m$ with $1\leq k+m \leq n$, the first $k$ cirquents $A_1,\cdots,A_k$ are of the form $B_{0}^{1}\add^{\mathfrak{b}_1} B_{1}^{1}$, \ldots, 
$B_{0}^{k}\add^{\mathfrak{b}_k} B_{1}^{k}$, the next $m$ cirquents are of the form $\ade^{\mathfrak{c}_1}x_1C_1(x_1), \cdots, \ade^{\mathfrak{c}_m}x_mC_m(x_m) $, and the remaining $n-k-m$ cirquents  are literals. Let $\cal M$ be a 
logical solution of $F$. Consider the work of $\cal M$ in the scenario where the environment does not move until $\cal M$ makes either (a)   the move $\mathfrak{b}_j.i$  for some $j\in\{1,\cdots,k\}$ and  $i\in\{0,1\}$,  or (b) 
the move $\mathfrak{c}_j.a$  for some $j\in\{1,\cdots, m\}$ and  $a\in\mathbb{N}$. At some point, $\cal M$ should indeed make such a move, for otherwise $F$ 
would be lost under an(y) interpretation which makes all of the literal cirquents $A_{k+m+1},\cdots,A_n$ false.
  In case (a), let $D$ be the result of replacing in $F$ every subcirquent of the form $X_0 \add^{\mathfrak{b}_j} X_1$  by $X_i$; in case (b), let $D$ be the result of replacing in $F$ every subcirquent of the form $ \ada^{\mathfrak{c}_j} y X (y)$  by $X (a)$.  With some analysis left to the reader, $\cal M$ can be seen to be a logical solution of $D$. Thus, $D$ is   valid. The rank of $D$ is smaller than that of $F$ and hence, by the induction hypothesis, $D$ is provable. But then so is $F$, because it follows from $D$ by Choosing. 

{\em Case 7}: $F=B_1\mlc\cdots\mlc B_m$ and $B_e$ are is as in Condition 7. 
 The  validity of $F$, of course, implies that 
$B_e$, as one of its $\mlc$-conjuncts, is also  valid. This rules out the possibility that $B_e$ is a nonlogical literal, because, as we observed earlier, a nonlogical literal cannot be   valid. Therefore  we are left with one of the following three possible subcases, corresponding to subconditions (1), (2) and (3) of Condition 7:

{\em Subcase 7.1}: $B_e$ is of the form $C_0\add^\mathfrak{c} C_1$. The argument given for Subcase 5.1 in the proof of Theorem 7.6 of \cite{cl16} goes through without any changes.

{\em Subcase 7.2}: $B_e$ is of the form $\ade^\mathfrak{c} xC(x)$. Let $\cal M$ be a logical solution of $F$.  Consider the work of $\cal M$ in the scenario where the environment does not move until $\cal M$ makes the move   $\mathfrak{c}.a$ for some constant $a$ (otherwise $F$ would be lost). 
  Let $D$ be the result of replacing  in $F$  every subcirquent of the form $ \ade^\mathfrak{c} y X(y)$ 
 by $X(a)$. Then, as in  Case 3, $\cal M$ can be seen to be a logical solution of $D$, meaning that  $D$ is   valid. The rank of $D$ is smaller than that of $F$ and hence, by the induction hypothesis, $D $ is provable. But then so is $F$, because it follows from $D $ by Chexists-choosing. 

{\em Subcase 7.3}: $B_e$ satisfies   Condition 6 in the role of $F$.  This 
case is very similar to Case 6 and, almost literally repeating our reasoning in the latter, we find that $F$ 
is provable. \end{proof}

\section{The decidability of $\mmm$} 
\begin{theorem}\label{theor}
$\mmm$   is decidable. Namely,  the algorithm   DECISION   described below accepts a closed cirquent $E$ if $\mmm\vdash E$ and rejects if $\mmm\not\vdash E$. 
\end{theorem}

\begin{proof} Consider an arbitrary closed cirquent $E$. The algorithm  DECISION  given below is a recursive one. It terminates because  every recursive call strictly decreases the rank of the cirquent that is being processed. This is how the algorithm acts on input $E$:

First, using the Purification algorithm,  DECISION   constructs the purification $F$ of $E$. Note   that, in  view of Theorem \ref{theo}, ``provable'' and ``valid'' can (and will) be used interchangeably; additionally, by clause 2 of Lemma \ref{pl}, so can be ``$E$'' and ``$F$'' when we talk about their provability or computability.
 As pointed out in the proof of Theorem \ref{theo}, $F$ should satisfy one of the eight Conditions listed  in that proof.

If $F$ is $\bot$ or a nonlogical literal as in Condition 0, it is  invalid, and we let  DECISION  reject $E$.   

If $F$ is $\top$ as in Condition 1, then it is valid, and we let  DECISION  accept  $E$.

Assume $F$ is $A_0 \add^\mathfrak{c}A_1$   as in Condition 2. Let $D_{i}$ ($i\in\{0,1\}$) be the result of  replacing in $F$ all subcirquents of the form $X_0\add^\mathfrak{c} X_1$ by $X_i$.    DECISION recursively calls itself  on $D_0$ and then on $D_1$ to figure out whether these cirquents are valid/provable. The ranks of both $D_0$ and $D_1$ are  smaller than the rank of  $F$ and hence, in view of clause 4 of Lemma \ref{pl}, smaller than the rank of $E$, as promised in the first paragraph of the present proof. 
 If at least one of $D_0,D_1$ turns out to be provable,    we let   DECISION   accept $E$,    because $F$ follows from either cirquent by (one or more applications of) Chor-choosing. Otherwise, if    both $D_0,D_1$ turn out to be invalid,  we let  DECISION  reject $E$ because,    as (in fact) observed within Case 2 of the proof of Theorem \ref{theo}, if $F$ was valid, then so would be either $D_0$ or $D_1$.

For the subsequent cases, we merely state how   DECISION  acts. A verification of the adequacy of the corresponding acceptance/rejection decisions is left to  the reader.

Assume $F$ is $\ade^\mathfrak{c} xA(x)$ as in Condition 3. Let $\{a_1,\ldots,a_s\}$ be all constants occurring in $F$, and let $a_{s+1}$ be a constant not occurring in $F$.\footnote{In fact, there is no need for considering $a_{s+1}$ unless $s=0$, but let us be generous. The same comment applies to our treatment of the cases of $F$ being as in Condition 6 or Condition 7(2).}   DECISION runs itself on each of the  cirquents $A(a_1),\cdots, A(a_{s+1})$. If all of these are rejected, then DECISION rejects $E$, otherwise it accepts $E$. 

Assume $F= A_0 \adc^\mathfrak{c} A_1$ is as in Condition 4. DECISION  runs itself on $A_0$ and then on $A_1$. If both are accepted, DECISION accepts $E$, otherwise it rejects $E$.  

Assume $F=\ada^\mathfrak{c} xA(x)$ is as in Condition 5. Let $c$ be a constant not occurring in $F$. DECISION  calls itself on $A(c)$ and generates the same acceptance/rejection decision for $E$ as the call does for $A(c)$.    

Assume $F=A_1\mld\cdots\mld A_n$    is as in Condition 6.  Let   $B_{0}^{1}\add^{\mathfrak{b}_1} B_{1}^{1}$, \ldots, 
$B_{0}^{k}\add^{\mathfrak{b}_k} B_{1}^{k}$   and   $\ade^{\mathfrak{c}_1}x_1C_1(x_1)$, \ldots, $\ade^{\mathfrak{c}_m}x_mC_1(x_m) $  be as in Case 6 of the proof of Theorem \ref{theo}. Let $a_1,\ldots,a_s$ be all constants occurring in $F$, and let $a_{s+1}$ be a constant not found in $F$.
  For $i\in\{1,\cdots,k\}$ and $j\in\{0,1\}$, let $D_{i}^{j}$  be the result of replacing  in $F$  all subcirquents of the form $X_0\add^{\mathfrak{b}_i} X_1$ by $X_j$. Further, for $i\in\{1,\cdots,m\}$ and $j\in\{ 1,\cdots, s+1 \}$, let $G_{i}^{j} $ be the result of replacing  in  $F$  all subcirquents of the form $\ada^{\mathfrak{c}_i}yX (y) $ by $X(a_j)$. DECISION calls itself on each of the  
cirquents  $D_{1}^{0},  \cdots,  D_{k}^{0}$, $D_{1}^{1},  \cdots,  D_{k}^{ 1}$,  \  $G_{1}^{1}, \cdots, G_{m}^{ 1}$, \ $\cdots$, \ $G_{1}^{s+1}, \cdots, G_{m}^{s+1}$. 
If all calls reject their arguments, then   DECISION rejects $E$, otherwise it accepts $E$. 

Finally, assume $F=B_1\mlc\cdots\mlc B_m$ and $B_e$ are   as in   Condition 7. 

Assume $A_e$ is of the form $C_0\add^\mathfrak{c} C_1$ as in Condition 7(1).  For $i\in\{0,1\}$, let $D_i$   be the result of replacing  in $F$  all subcirquents of the form  $X_0\add^\mathfrak{c} X_1$ by $X_i$.  DECISION runs itself on $D_0$ and   $D_1$. If both cirquents are rejected, then DECISION rejects $E$, otherwise it accepts $E$.   

Assume $A_e$ is of the form  $\ade^\mathfrak{c} xC(x)$ as in Condition 7(2). Let $a_1,\ldots,a_s$ be all constants occurring in $F$, and let $a_{s+1}$ be a constant not found in $F$.  For each $j\in\{1,\ldots,s+1\}$, let $D_j$ be the result of replacing in $F$ all subcirquents of the form $\ade^\mathfrak{c}yX(y)$ by $X(a_j)$.  DECISION runs itself on each of the arguments $D_1,\ldots,D_{s+1}$. If all $s+1$ cirquents are rejected,   DECISION rejects $E$, otherwise it accepts $E$.   

Assume $A_e$ is as in Condition 7(3). With $A_e$ in the role of $F$,   DECISION acts exactly as in the above case of $F$ satisfying Condition 6.  
\end{proof}

\end{document}